\documentclass[journal,web]{ieeecolor}
\usepackage{multirow}
\usepackage{float}
\usepackage{generic}
\usepackage{cite}
\usepackage{amsmath,amssymb,amsfonts}
\usepackage{algorithmic}
\usepackage{graphicx}
\usepackage{algorithm,algorithmic}
\usepackage{soul} 
\usepackage{xcolor}
\usepackage{tikz} 
\usepackage{color}
\usepackage{graphicx}
\graphicspath{{figures/}{./}}
\usepackage{textcomp}
\usepackage{trimclip}
\usepackage[mathscr]{euscript}
\usepackage{array}
\usepackage{eqparbox}
\usepackage{url}
\usepackage{relsize}
\usepackage{dsfont}
\usepackage{nopageno}

\newtheorem{theorem}{Theorem}
\newtheorem{lemma}[theorem]{Lemma}
\newtheorem{proposition}[theorem]{Proposition}

\newtheorem{definition}{Definition}
\newtheorem{corollary}[theorem]{Corollary}
\newtheorem{remark}{Remark}

\begin{document}

\bstctlcite{}
\title{Data-Driven Model Order Reduction for T-Product-Based Dynamical Systems}
\author{Shenghan Mei, Ziqin He, Yidan Mei, Xin Mao, Anqi Dong, \\ Ren Wang, \textit{Member}, \textit{IEEE}, and Can Chen, \textit{Member}, \textit{IEEE}
\thanks{Shenghan Mei and Ziqin He contributed equally to this work.}
\thanks{Shenghan Mei, Ziqin He, and Yidan Mei are with the Department of Mathematics, University of North Carolina at Chapel Hill, Chapel Hill, NC 27599, USA (emails: zhe21@unc.edu; ymei@unc.edu; shmei@unc.edu).}
\thanks{Xin Mao is with the School of Data Science and Society, University of North Carolina at Chapel Hill, Chapel Hill, NC 27599, USA (email: xinm@unc.edu).}
\thanks{Anqi Dong is with the Division of Decision and Control Systems and the Department of Mathematics, KTH Royal Institute of Technology, SE-100 44 Stockholm, Sweden (email: anqid@kth.se).}
\thanks{Ren Wang is with the Department of Electrical and Computer Engineering, Illinois Institute of Technology, Chicago, IL 60616, USA (e-mail:
rwang74@iit.edu).}
\thanks{Can Chen is with the School of Data Science and Society, the Department of Mathematics, and the Department of Biostatistics,  University of North Carolina at Chapel Hill, Chapel Hill, NC 27599, USA (email: canc@unc.edu).}%
}

\maketitle
\thispagestyle{empty}
%\thispagestyle{empty}
%%%%%%%%%%%%%%%%%%%%%%%%%%%%%%%%%%%%%%%%%%%%%%%%%%%%%%%%%%%%%%%%%%%%%%%%%%%%%%%%
\begin{abstract}
Model order reduction plays a crucial role in simplifying complex systems while preserving their essential dynamic characteristics, making it an invaluable tool in a wide range of applications, including robotic systems, signal processing, and fluid dynamics. However, traditional model order reduction techniques like balanced truncation are not designed to handle tensor data directly and instead require unfolding the data, which may lead to the loss of important higher-order structural information. In this article, we introduce a novel framework for data-driven model order reduction of T-product-based dynamical systems (TPDSs), which are often used to capture the evolution of third-order tensor data such as images and videos through the T-product. Specifically, we develop advanced T-product-based techniques, including T-balanced truncation, T-balanced proper orthogonal decomposition, and the T-eigensystem realization algorithm for input-output TPDSs by leveraging the unique properties of T-singular value decomposition. We demonstrate that these techniques offer significant memory and computational savings while achieving reduction errors that are  comparable to those of conventional methods. The effectiveness of the proposed framework is further validated through synthetic and real-world examples.
\end{abstract}


\section{Introduction}
Model order reduction serves as a fundamental methodology for simplifying complex dynamical systems while preserving their essential dynamic behavior \cite{baur2014model,lassila2014model,chinesta2011short,chinesta2016model,touze2021model,sandberg2009model}. By reducing the state dimensionality of a system, model order reduction facilitates more efficient simulation, analysis, and control, particularly for systems with large state spaces or high computational complexity. Traditional model order reduction techniques, including balanced truncation \cite{moore1981principal,gugercin2004survey,sandberg2004balanced,sandberg2010extension}, proper orthogonal decomposition \cite{berkooz1993proper,weiss2019tutorial}, and the eigensystem realization algorithm \cite{juang1985eigensystem,pappa1993consistent}, have been extensively studied for linear and nonlineaar systems. The core idea behind these methods is to construct a Hankel matrix (or its variants) \cite{widom1966hankel,tyrtyshnikov1994bad}  that encodes the system's input-output behavior over time and apply matrix singular value decomposition to identify the most important dynamic modes. Model order reduction has achieved notable success across various applications, including robotics \cite{goury2018fast,bolduc1998review}, fluid dynamics \cite{carlberg2013gnat,lassila2014model}, structural health monitoring \cite{rosafalco2021online,bolandi2019novel},  signal processing \cite{vaseghi2008advanced,baraniuk2010low}, and coarse-graining in multiscale modeling \cite{gorban2006model,saunders2013coarse}. However, these classical approaches are inherently tailored for systems with vector variables and do not readily extend to tensor-based dynamical systems, where the state, input, and output variables are represented as higher-order tensors.

Numerous real-world systems, such as those found in image and video processing, biological networks, and chemical reactions, exhibit intricate multidimensional relationships that are more naturally and effectively captured using tensor-based dynamical systems \cite{chen2024tensor,favier2023tensor,zhao2013kernelization,10759686,chen2022explicit}. We are particularly interested in T-product-based dynamical systems (TPDSs), a special class of tensor-based systems where the system evolution  is governed by the T-product between a third-order dynamic tensor and a third-order state tensor. Originally introduced by Hoover et al. \cite{hoover2021new}, TPDSs extend the classical linear system framework to third-order tensor data, offering a powerful tool for modeling complex interactions in structured data such as images, videos, and spatial-temporal signals \cite{koniusz2021tensor,zhou2012compact,alawode2025learning}.  The system-theoretic properties of TPDSs, such as stability, controllability, and observability, have been systematically investigated, along with optimal control techniques like state feedback control \cite{chen2024tensor, hoover2021new}. Additionally, data-driven approaches for the analysis and control of TPDSs have also been actively explored, which enhances the applicability to real-world, multidimensional data  \cite{mao2024data,he2025data}. Nevertheless, the development of data-driven model order reduction techniques for TPDSs remains an open area.

To apply the classical model order reduction approaches to input-output TPDSs,  the tensors must first be vectorized, resulting in an exceedingly high-dimensional system representation, where the number of states and parameters grows cubically with the dimensionality of the tensors. More critically, this vectorization process disrupts the inherent multidimensional structures of the data, such as correlations, redundancies, and spatial-temporal dependencies, making it impossible to directly exploit these rich structural patterns for efficient representation and computation \cite{chen2021multilinear}. To address these challenges, T-product algebra, including  T-singular value decomposition, provides a natural and powerful foundation for developing data-driven model order reduction techniques tailored  for TPDSs. Indeed, several studies have demonstrated that leveraging the unique properties of the T-product can significantly enhance the efficiency of representation and computation for TPDSs \cite{chen2024tensor,mao2024data,he2025data}.

The goal of this article is to develop a novel and efficient model order reduction framework for input-output TPDSs. Specifically, we explore different flavors of balanced truncation methods: (i) T-balanced truncation, which involves solving tensor Lyapunov equations to obtain controllability and observability Gramians, followed by the computation of the balancing transform; (ii) T-balanced proper orthogonal decomposition, which directly and more efficiently constructs the balancing transform from forward/adjoint snapshots; and (iii) the T-eigensystem realization algorithm, which relies on snapshots from the forward model and can also be applied to experimental data. Similar to the original approaches, the core idea behind these methods is to construct a Hankel tensor (or its variants) that captures the system’s input-output behavior over time and apply T-singular value decomposition to identify the most significant dynamic modes by truncating singular tuples. Moreover, we present detailed analyses of the memory and computational complexity for the proposed methods and compare them with their corresponding standard methods using numerical examples.

Data-driven model order reduction for input-output TPDSs has significant potential applications in image and video compression, addressing the growing need for efficient methods to manage and process large-scale visual data \cite{zheng2023approximation,yin2018multiview,tarzanagh2018fast}. In image and video dynamics, data are often represented as time-series high-dimensional third-order tensors, which can be effectively modeled using TPDSs. However, the complexity and size of these tensors can result in massive volumes of information that are computationally expensive to store, process, and transmit. As image and video resolutions increase and video frame rates become higher, the amount of data generated grows exponentially. Model order reduction techniques for TPDSs provide a promising solution to these challenges by reducing the dimensionality of the underlying system, effectively compressing the data while preserving the most critical features and dynamics.  This approach not only facilitates efficient data storage and transmission but also accelerates downstream tasks such as filtering, object tracking, and motion analysis.

This article is organized into five sections. Section \ref{sec:prelim} provides an overview of the T-product with its associated algebra and traditional model order reduction methods. In Section \ref{sec:meth}, we introduce advanced T-product-based model order reduction techniques, including T-balanced truncation, T-balanced proper
orthogonal decomposition, and the T-eigensystem realization
algorithm for input-output TPDSs by leveraging  T-singular value decomposition. Section \ref{sec:num} presents numerical examples, including both synthetic data and a case study on image data, to illustrate the proposed methods. Finally, we conclude with
a discussion of future directions in Section \ref{sec:con}.

\section{Preliminary}\label{sec:prelim}

\subsection{Tensors}
As a generalization of vectors and matrices, tensors are multidimensional arrays that extend the concept of linear algebra to higher orders \cite{chen2024tensor, kolda2009tensor, ragnarsson2012block,pickard2024kronecker}. The order of a tensor is determined by the number of its dimensions, and each dimension is referred to as a mode. This article particularly deals with third-order tensors, which we denote by $\mathscr{A} \in \mathbb{R}^{n \times m \times s}$. We begin by introducing the T-product, a fundamental operation for third-order tensors that facilitates their multiplication via circular convolution \cite{kilmer2013third, kilmer2011factorization, zhang2016exact,miao2020generalized,chang2022t}.
\begin{definition}
The T-product between two third-order tensor $\mathscr{A}\in\mathbb{R}^{n\times m\times s}$ and $\mathscr{B}\in\mathbb{R}^{m\times h\times s}$ is defined as 
\begin{equation}
\mathscr{A}\circledast\mathscr{B} = \mu^{-1}\big(\xi(\mathscr{A})\mu(\mathscr{B})\big) \in\mathbb{R}^{n\times h\times s},
\end{equation}
where  $\xi(\cdot)$ and $\mu(\cdot)$ denote the block circulant and unfold operators defined as
\begin{align*}\label{eq:bcirc}
\xi(\mathscr{A})&= \begin{bmatrix}
\mathscr{A}_{::1} & \mathscr{A}_{::s} & \cdots & \mathscr{A}_{::2}\\
\mathscr{A}_{::2} & \mathscr{A}_{::1} & \cdots & \mathscr{A}_{::3}\\
\vdots & \vdots & \ddots & \vdots\\
\mathscr{A}_{::s} & \mathscr{A}_{::(s-1)} & \cdots & \mathscr{A}_{::1}
\end{bmatrix}\in\mathbb{R}^{ns\times ms},\\
\mu(\mathscr{B}) &= \begin{bmatrix}
\mathscr{B}_{::1}^\top &
\mathscr{B}_{::2}^\top&
\cdots&
\mathscr{B}_{::s}^\top
\end{bmatrix}^\top\in\mathbb{R}^{ms\times h},
\end{align*}
respectively. Note that $\mathscr{A}_{::j}$ refers to the frontal slices, which are matrices obtained by fixing the index of the third mode while allowing the indices of the first two modes to vary.
\end{definition}

\subsubsection{Basic T-Product Algebra} Many fundamental linear algebra notions and operations, including identity, transpose, inverse, orthogonality, and block matrices, can be generalized to third-order tensors through the T-product \cite{kilmer2013third, kilmer2011factorization, zhang2016exact}.

\begin{definition}
The T-identity tensor $\mathscr{I}\in\mathbb{R}^{n\times n\times s}$ is defined such that its first frontal slice $\mathscr{I}_{::1}$ is the identity matrix, while all other frontal slices are zero matrices.
\end{definition}

\begin{definition}
The T-transpose of a third-order tensor $\mathscr{A} \in \mathbb{R}^{n \times m \times s}$, denoted as $\mathscr{A}^\top\in\mathbb{R}^{m\times n\times s}$, is obtained by transposing each frontal slice and reversing the order of the transposed slices from the second to the last. If $\mathscr{A}=\mathscr{A}^\top$, it is called T-symmetric. 
\end{definition}

\begin{definition}
The T-inverse of a third-order tensor $\mathscr{A} \in \mathbb{R}^{n \times n \times s}$, denoted as $\mathscr{A}^{-1}\in \mathbb{R}^{n \times n \times r}$, satisfies $\mathscr{T} \circledast \mathscr{T}^{-1} = \mathscr{T}^{-1} \circledast \mathscr{T} = \mathscr{I}$. The left and
right T-inverses of $\mathscr{A}$ can be defined similarly. 
\end{definition}

\begin{definition}
A third-order tensor $\mathscr{A} \in \mathbb{R}^{n \times n \times s}$ is called T-orthogonal if $\mathscr{A} \circledast \mathscr{A}^\top = \mathscr{A}^\top \circledast \mathscr{A} = \mathscr{I}$.   
\end{definition}

\begin{definition}
Given two third-order tensors $\mathscr{A}\in\mathbb{R}^{n\times m\times s}$ and $\mathscr{B}\in\mathbb{R}^{n\times h\times s}$, the row block tensor, denoted as $[\mathscr{A} \text{ }\mathscr{B}]\in\mathbb{R}^{n\times (m+h)\times s}$, is formed by concatenating $\mathscr{A}$ and $\mathscr{B}$ along the second mode. On the other other, given two  third-order tensors $\mathscr{A}\in\mathbb{R}^{n\times m\times s}$ and $\mathscr{B}\in\mathbb{R}^{h\times m\times s}$, the column block tensor, denoted as $[\mathscr{A} \text{ }\mathscr{B}]^\top\in\mathbb{R}^{(n+h)\times m\times s}$, is formed by concatenating $\mathscr{A}$ and $\mathscr{B}$ along the first mode.
\end{definition}

Note that we use the same notations for matrix and T-product operations (e.g., transpose, inverse, and block matrices/tensors) for convenience. Moreover, all operations mentioned above can be computed or verified through $\xi(\cdot)$. For instance, $\mathscr{I}\in\mathbb{R}^{n\times n\times s}$ is a T-identity tensor if and only if $\xi(\mathscr{I})\in\mathbb{R}^{ns\times ns}$ is an identity matrix. 

\subsubsection{T-Singular Value Decomposition}
Matrix singular value decomposition (SVD) can  be extended to
third-order tensors using the T-product, which plays a salient role in data-driven model order reduction for TPDSs \cite{kilmer2013third, braman2010third}. 

\begin{definition}
The T-singular value decomposition (T-SVD) of a third-order tensor $\mathscr{A}\in\mathbb{R}^{n\times m\times s}$ is defined as
\begin{equation}
\mathscr{A}=\mathscr{U}\circledast\mathscr{S}\circledast\mathscr{U}^{\top},
\end{equation}
where $\mathscr{U}\in\mathbb{R}^{n\times n\times s}$ and $\mathscr{V}\in\mathbb{R}^{m\times m\times s}$ are T-orthogonal, and $\mathscr{S}\in\mathbb{R}^{n\times m\times s}$ is an F-rectangle diagonal tensor (i.e., each frontal slice of $\mathscr{S}$ is a rectangle diagonal matrix). The vectors $\mathscr{S}_{jj:} \in \mathbb{R}^{s}$ are referred to as the singular tuples of $\mathscr{A}$.
\end{definition}

T-SVD can be derived by applying the discrete Fourier transform and matrix SVD.  Specifically, for a third-order tensor $\mathscr{A} \in \mathbb{R}^{n \times m \times s}$, the process begins by performing the discrete Fourier transform on $\xi(\mathscr{A})$, which yields
\begin{align*}
\mathcal{F}\{\xi(\mathscr{A})\} =&(\textbf{F}_n \otimes \textbf{I}_s)  \xi(\mathscr{A})  (\textbf{F}_n^* \otimes \textbf{I}_s)\\ 
=& \texttt{blkdiag}(\textbf{A}_1, \textbf{A}_2, \dots, \textbf{A}_s),
\end{align*}
where  $\texttt{blkdiag}$ is the MATLAB block diagonal function, $\otimes$ represents the Kronecker product, the superscript $*$ denotes the conjugate transpose, $\textbf{I}_r \in \mathbb{R}^{s \times s}$ is the identity matrix, and $\textbf{F}_n \in \mathbb{C}^{n \times n}$ is the discrete Fourier transform matrix 
\begin{equation*}
\textbf{F}_n = \frac{1}{\sqrt{n}} \begin{bmatrix}
1 & 1 & 1 & \cdots & 1 \\
1 & \omega & \omega^2 & \cdots & \omega^{n-1} \\
\vdots & \vdots & \vdots & \ddots & \vdots \\
1 & \omega^{n-1} & \omega^{2(n-1)} & \cdots & \omega^{(n-1)^2}
\end{bmatrix}
\end{equation*}
with $\omega = \exp{\left(\frac{-2\pi i}{n}\right)}$. Next, the SVD of each diagonal block matrix $\textbf{A}_j \in \mathbb{R}^{n \times n}$ is computed as $\textbf{A}_j = \textbf{U}_j \textbf{S}_j \textbf{V}_j^{\top}$ for $j = 1, 2, \ldots, s$. Finally, the left T-orthogonal  factor tensor $\mathscr{U}$ in the T-SVD of $\mathscr{A}$ can be obtained as
\begin{equation*}
\mathscr{U} = \xi^{-1} \big( (\textbf{F}_n^* \otimes \textbf{I}_s)  \texttt{blkdiag}(\textbf{U}_1, \textbf{U}_2, \ldots, \textbf{U}_s)  (\textbf{F}_n \otimes \textbf{I}_s) \big).
\end{equation*}
The right factor tensor $\mathscr{V}$ and the F-rectangle diagonal tensor $\mathscr{S}$ can be constructed in a similar manner. It is worth noting that the T-SVD of $\mathscr{A}$ differs from the matrix SVD of $\xi(\mathscr{A})$ since the block circulant matrix $\xi(\mathscr{S})$ is not diagonal.
Furthermore, T-eigenvalue decomposition (T-EVD) can be defined and computed similarly.
\begin{definition}
The T-eigenvalue decomposition (T-EVD) of a third-order tensor $\mathscr{A} \in \mathbb{R}^{n \times n \times s}$ is defined as
\begin{equation}
\mathscr{A} = \mathscr{U} \circledast \mathscr{D} \circledast \mathscr{U}^{-1},
\end{equation}
where $\mathscr{U} \in \mathbb{C}^{n \times n \times s}$ and $\mathscr{D} \in \mathbb{C}^{n \times n \times s}$ is a F-diagonal tensor. The tubes of $\mathscr{D}$, denoted as $\mathscr{D}_{jj:} \in \mathbb{C}^{s}$, are referred to as the eigentuples of $\mathscr{A}$.
\end{definition}

\subsection{Linear Model Order Reduction Techniques}

Various model order reduction techniques have been developed for linear dynamical systems, including balanced truncation, balanced proper orthogonal decomposition, and the eigensystem realization algorithm. A discrete-time input-output linear time-invariant system is represented as
\begin{align}\label{eq:linear}
\begin{cases}
\textbf{x}(t+1) = \textbf{A}\textbf{x}(t) + \textbf{B} \textbf{u}(t)\\
\textbf{y}(t) = \textbf{C} \textbf{x}(t)
\end{cases},   
\end{align}
where $\textbf{A}\in\mathbb{R}^{n\times n}$ is the state transition matrix, $\textbf{B}\in\mathbb{R}^{n\times m}$ is the control matrix, and $\textbf{C}\in\mathbb{R}^{l\times n}$ is the output matrix. Here, $\textbf{x}(t)\in\mathbb{R}^n$ is the state, $\textbf{u}(t)\in\mathbb{R}^m$ is the control input, and $\textbf{y}(t)\in\mathbb{R}^l$ is the output. The goal of model order reduction is to derive a reduced system with $\textbf{A}_\text{red} \in\mathbb{R}^{r\times r}$, $\textbf{B}_\text{red} \in\mathbb{R}^{r\times m}$, and $\textbf{C}_\text{red} \in\mathbb{R}^{l\times r}$ for $r< n$, which retains the essential dynamic behavior of the original system while reducing its state dimensionality.

%to find transformation matrices  $\textbf{P}\in\mathbb{R}^{n\times r}$ and $\textbf{Q}\in\mathbb{R}^{n\times r}$ with $r\ll n$ such that the \textbf{Q}^\top\textbf{A}\textbf{P}= \textbf{Q}^\top\textbf{B}= \textbf{C}\textbf{P}

\subsubsection{Balanced Truncation}
Balanced truncation (BT) is a widely used model order reduction technique for linear systems \cite{moore1981principal}. The procedure begins with the computation of the controllability Gramian $\textbf{W}_\text{c}$ and observability Gramian $\textbf{W}_\text{o}$ from solving the Lyapunov equations of the linear system \eqref{eq:linear}
\begin{align*}
    \textbf{W}_\text{c}-\textbf{A}\textbf{W}_\text{c}\textbf{A}^\top &= \textbf{B}\textbf{B}^\top,\\
    \textbf{A}^\top\textbf{W}_\text{o}\textbf{A} - \textbf{W}_\text{o} &= -\textbf{C}^\top\textbf{C}.
\end{align*}
Once the Gramians are obtained, Cholesky factorizations are performed to obtain
\begin{equation*}
    \textbf{W}_\text{c} = \textbf{Z}_\text{c}\textbf{Z}_\text{c}^\top\quad \text{and} \quad\textbf{W}_\text{o} = \textbf{Z}_\text{o}\textbf{Z}_\text{o}^\top,
\end{equation*}
where $\textbf{Z}_\text{c}\in\mathbb{R}^{n\times p}$ and $\textbf{Z}_\text{o}\in\mathbb{R}^{n\times q}$.
A Hankel matrix $\textbf{H}=\textbf{Z}_\text{o}^\top\textbf{Z}_\text{c}\in\mathbb{R}^{q\times p}$ is then constructed, and its truncated SVD is computed as $\textbf{H}\approx\textbf{U}\textbf{S}\textbf{V}^\top$ where $\textbf{U}\in\mathbb{R}^{q\times r}$, $\textbf{S}\in\mathbb{R}^{r\times r}$, and $\textbf{V}\in\mathbb{R}^{p\times r}$. This decomposition leads to the derivation of the two transformation matrices
$
    \textbf{P} = \textbf{Z}_\text{c}\textbf{V}\textbf{S}^{-\frac{1}{2}}\in\mathbb{R}^{n\times r}$ and $\textbf{Q} = \textbf{Z}_\text{o}\textbf{U}\textbf{S}^{-\frac{1}{2}}\in\mathbb{R}^{n\times r}$.
The reduced system is then computed as $\textbf{A}_\text{red}=\textbf{Q}^\top\textbf{A}\textbf{P}\in\mathbb{R}^{r\times r}$, $\textbf{B}_\text{red}=\textbf{Q}^\top\textbf{B}\in\mathbb{R}^{r\times m}$, and $\textbf{C}_\text{red}=\textbf{C}\textbf{P}\in\mathbb{R}^{l\times r}$. BT identifies and eliminates less significant states that have minimal impact on the system's input-output response, which yields a reduced model that preserves the critical dynamic behaviors of the original system.

\subsubsection{Balanced Proper Orthogonal Decomposition} 
The method of balanced
proper orthogonal decomposition (BPOD) was first proposed by Rowley \cite{rowley2017model} to address model order reduction challenges in high-dimensional input/output spaces, particularly within the context of fluid mechanics. This approach involves computing the balancing transformation from a collection of snapshots obtained through impulse response simulations. Suppose that the control matrix \textbf{B} can be rewritten as $\textbf{B}=[\textbf{b}_1 \text{ }\textbf{b}_2 \text{ }\cdots\text{ }\textbf{b}_m]$. The state responses for the linear system \eqref{eq:linear} over $t= 0,1,\dots,T$ can be constructed as
\begin{equation*}
    \textbf{X}^{(j)} = \begin{bmatrix}
        \textbf{b}_j & \textbf{A}\textbf{b}_j & \cdots & \textbf{A}^T\textbf{b}_j
    \end{bmatrix}
\end{equation*}
for $j=1,2,\dots,m$. These snapshots are then arranged as
\begin{equation*}
    \textbf{X} = \begin{bmatrix}
        \textbf{X}^{(1)} &  \textbf{X}^{(2)} & \cdots &  \textbf{X}^{(m)}
    \end{bmatrix}\in\mathbb{R}^{n\times mT},
\end{equation*}
and the empirical controllability Gramian is computed as $\textbf{W}_\text{c} = \textbf{X}\textbf{X}^\top$ \cite{lall1999empirical}. Similarly,
the procedure proceeds for constructing the output snapshot $\textbf{Y}\in\mathbb{R}^{n\times lL}$ and the associated empirical observability Gramian $\textbf{W}_\text{o}$ from the simulations of the adjoint system of \eqref{eq:linear} over $t= 0,1,\dots,L$. Finally, a Hankel matrix $\textbf{H}$ is constructed as $\textbf{H}=\textbf{Y}^\top\textbf{X}$, and the remaining steps follow exactly as those in balanced truncation. BPOD avoids the need to compute the controllability and observability Gramians by solving the Lyapunov equations, enhancing computational efficiency for large-scale systems.

\subsubsection{Eigensystem Realization Algorithm} 
The eigensystem realization algorithm (ERA), first introduced in \cite{juang1985eigensystem}, is a model identification and reduction tool for linear systems. The core idea involves constructing a generalized Hankel matrix using impulse response simulations or experiments, without needing access to the underlying linear system \eqref{eq:linear}. Specifically, snapshots are collected and arranged as
\begin{equation*}
\textbf{H}=
\begin{bmatrix}
    \textbf{Z}_0 & \textbf{Z}_1 & \cdots & \textbf{Z}_T\\
    \textbf{Z}_1 & \textbf{Z}_2 & \cdots & \textbf{Z}_{T+1}\\
    \vdots & \vdots & \ddots & \vdots\\
    \textbf{Z}_L & \textbf{Z}_{L+1} & \cdots & \textbf{Z}_{L+T}\\
\end{bmatrix}\in\mathbb{R}^{l(L+1)\times m(T+1)},
\end{equation*}
where $\textbf{Z}_j=\textbf{C}\textbf{A}^j\textbf{B}\in\mathbb{R}^{l\times m}$ are called the Markov parameters, and $T+L+2$ is the number of snapshots. Suppose that the truncated SVD of \textbf{H} is computed as $\textbf{H}\approx\textbf{U}\textbf{S}\textbf{V}^\top$ where $\textbf{U}\in\mathbb{R}^{l(L+1)\times r}$, $\textbf{S}\in\mathbb{R}^{r\times r}$, and $\textbf{V}\in\mathbb{R}^{m(T+1)\times r}$. The reduced system is then constructed as 
$\textbf{A}_\text{red}= \textbf{S}^{-\frac{1}{2}}\textbf{U}^\top\hat{\textbf{H}}\textbf{V}\textbf{S}^{-\frac{1}{2}}\in\mathbb{R}^{r\times r}$, $\textbf{B}_{\text{red}} = \textbf{S}^{-\frac{1}{2}}\textbf{U}^\top \textbf{H}_{:(1:m)}\in\mathbb{R}^{r\times m}$, and $\textbf{C}_{\text{red}} = \textbf{H}_{(1:l):}\textbf{U}\textbf{S}^{-\frac{1}{2}} \in\mathbb{R}^{l\times r}$, where 
\begin{equation*}
\hat{\textbf{H}}=
\begin{bmatrix}
    \textbf{Z}_1 & \textbf{Z}_2 & \cdots & \textbf{Z}_{T+1}\\
    \textbf{Z}_2 & \textbf{Z}_3 & \cdots & \textbf{Z}_{T+2}\\
    \vdots & \vdots & \ddots & \vdots\\
    \textbf{Z}_{L+1} & \textbf{Z}_{L+2} & \cdots & \textbf{Z}_{L+T+1}\\
\end{bmatrix}\in\mathbb{R}^{l(L+1)\times m(T+1)}.
\end{equation*}
ERA has been shown to be theoretically equivalent to BPOD for discrete-time systems, but with significantly lower computational cost \cite{ma2011reduced}.

\section{Data-driven Model Order Reduction for TPDSs}\label{sec:meth}

We consider a discrete-time input-output T-product-based dynamical system (TPDS) represented as
\begin{align}\label{eq:tpds}
\begin{cases}
\mathscr{X}(t+1) = \mathscr{A}\circledast\mathscr{X}(t) + \mathscr{B} \circledast\mathscr{U}(t)\\
\mathscr{Y}(t) = \mathscr{C} \circledast\mathscr{X}(t)
\end{cases},   
\end{align}
where $\mathscr{A}\in\mathbb{R}^{n\times n\times s}$ is the state transition tensor, $\mathscr{B}\in\mathbb{R}^{n\times m\times s}$ is the control tensor, and $\mathscr{C}\in\mathbb{R}^{l\times n\times s}$ is the output tensor. Here, $\mathscr{X}(t)\in\mathbb{R}^{n\times h\times s}$ is the state, $\mathscr{U}(t)\in\mathbb{R}^{m\times h\times s }$ is the control input, and $\mathscr{Y}(t)\in\mathbb{R}^{l\times h\times s}$ is the output. The objective of model order reduction is to obtain a reduced-order input-output TPDS that preserves the key dynamic characteristics of the original system while decreasing its state dimensionality.

The input-output TPDS \eqref{eq:tpds} can be reformulated as a linear system using the block circulant operation, resulting in the following representation
\begin{align}\label{eq:unfoldtpds}
\begin{cases}
\mu(\mathscr{X}(t+1)) = \xi(\mathscr{A})\mu(\mathscr{X}(t)) + \xi(\mathscr{B} )\mu(\mathscr{U}(t))\\
\mu(\mathscr{Y}(t)) = \xi(\mathscr{C}) \mu(\mathscr{X}(t))
\end{cases}.   
\end{align}
This linear representation allows for the application of standard linear model order reduction techniques, such as BT, BPOD, and  ERA. However, it overlooks the inherent structure of block circulant matrices and the potential advantages of advanced tensor algebra methods like T-SVD, which could enhance both computational and memory efficiency. We therefore introduce three extensions of data-driven model order reduction techniques for input-output TPDSs, including T-balanced truncation (T-BT), T-balanced proper orthogonal decomposition (T-BPOD), and the T-eigensystem realization algorithm (T-ERA). These techniques hold potential for applications like capturing image and video dynamics by efficiently reducing dimensionality while retaining essential features.

\subsection{T-Balanced Truncation}
Controllability and observability are fundamental concepts in BT-based model order reduction. For input-output TPDSs, both properties have been formally introduced in \cite{mao2024data, chen2024tensor}, with formulations that mirror those in classical linear systems theory.  Similar to BT, the first step in T-BT involves solving for the controllability and observability Gramians of the input-output TPDS \eqref{eq:tpds}, which can  be used to determine the system's controllability and observability.
\begin{definition}
    The controllability and observability Gramians of the input-output TPDS \eqref{eq:tpds} are defined as
    \begin{equation}
    \begin{split}
        \mathscr{W}^{\text{c}} &= \sum_{k=0}^\infty \mathscr{A}^k\circledast\mathscr{B}\circledast\mathscr{B}^\top \circledast (\mathscr{A}^\top)^k\in\mathbb{R}^{n\times n\times s},\\
        \mathscr{W}^{\text{o}} &= \sum_{k=0}^\infty (\mathscr{A}^\top)^k\circledast\mathscr{C}^\top\circledast\mathscr{C} \circledast \mathscr{A}^k\in\mathbb{R}^{n\times n\times s},
        \end{split}
    \end{equation}
respectively. 
\end{definition}

\begin{proposition}\label{prop:control}
    The input-output TPDS \eqref{eq:tpds} is controllable (observable) if and only if all eigentuples of the controllability Gramian $\mathscr{W}^{\text{c}}$ (observability Gramian $\mathscr{W}^{\text{o}}$) in the Fourier domain contain positive entries. 
\end{proposition}
\begin{proof}
    Based on the controllability condition for linear systems, the unfolded representation \eqref{eq:unfoldtpds} is controllable if and only if the controllability Gramian
    \begin{equation*}
        \textbf{W}^{\text{c}} = \sum_{k=0}^\infty \xi(\mathscr{A})^k\xi(\mathscr{B})\xi(\mathscr{B})^\top (\xi(\mathscr{A})^\top)^k\in\mathbb{R}^{ns\times ns}
    \end{equation*}
is positive definite. By leveraging the properties of the block circulant operator, it can be shown that $\textbf{W}^{\text{c}} = \xi(\mathscr{W}^{\text{c}})$, and then the eigenvalues of $\textbf{W}^{\text{c}}$ correspond to the combined entries of the eigentuples of $\mathscr{W}^{\text{c}}$ in the Fourier domain. Therefore, the unfolded representation \eqref{eq:unfoldtpds} is controllable if and only if all eigentuples of the controllability Gramian $\mathscr{W}^{\text{c}}$ in the Fourier domain contain positive entries. The result for the input-output TPDS \eqref{eq:tpds} follows immediately. A similar argument can be established for observability.
\end{proof}

According to the definition of T-EVD, the eigentuples of $\mathscr{W}^\text{c}$ in the Fourier domain can be computed from the eigenvalues of its diagonal block matrices under the Fourier transform, i.e., $\mathcal{F}(\xi(\mathscr{W}^\text{c}))=\texttt{blkdiag}(\textbf{W}^\text{c}_1, \textbf{W}^\text{c}_2,\dots, \textbf{W}^\text{c}_s)$ (similarly for $\mathscr{W}^\text{o}$). Moreover, the controllability and observability Gramians can be obtained by solving the T-Lyapunov equations associated with the input-output TPDS \eqref{eq:tpds}.

\begin{proposition}\label{prop:lyapunov}
    The controllability and observability Gramians $\mathscr{W}^\text{c}$ and $\mathscr{W}^\text{o}$ of the input-output TPDS \eqref{eq:tpds} can be obtained by solving the associated T-Lyapunov equations   defined as
    \begin{equation}
    \begin{split}
        \mathscr{W}^\text{c}-\mathscr{A}\circledast\mathscr{W}^\text{c}\circledast\mathscr{A}^\top &= \mathscr{B}\circledast\mathscr{B}^\top,\\
    \mathscr{A}^\top\circledast\mathscr{W}^\text{o}\circledast\mathscr{A} - \mathscr{W}^\text{o} &= -\mathscr{C}^\top\circledast\mathscr{C},
    \end{split}
    \end{equation}
respectively. 
\end{proposition}
\begin{proof}
    According to Proposition \ref{prop:control}, the controllability and observability Gramians $\textbf{W}^{\text{c}}$ and $\textbf{W}^{\text{o}}$ of the unfolded representation \eqref{eq:unfoldtpds} satisfy $\textbf{W}^{\text{c}} = \xi(\mathscr{W}^{\text{c}})$ and $\textbf{W}^{\text{o}} = \xi(\mathscr{W}^{\text{o}})$. In addition, both matrices can be solved from the Lyapunov equations associated with the unfolded representation, i.e., 
    \begin{align*}
        \textbf{W}^{\text{c}}-\xi(\mathscr{A})\textbf{W}^{\text{c}}\xi(\mathscr{A})^\top &= \xi(\mathscr{B})\xi(\mathscr{B})^\top,\\
        \xi(\mathscr{A})^\top\textbf{W}^{\text{o}}\xi(\mathscr{A}) -\textbf{W}^{\text{o}} &= -\xi(\mathscr{C})^\top\xi(\mathscr{C}).
    \end{align*}
Replace $\textbf{W}^{\text{c}}$ and $\textbf{W}^{\text{o}}$ with $\xi(\mathscr{W}^{\text{c}})$ and $\xi(\mathscr{W}^{\text{o}})$, and by the definition
of the block circulant operation, it follows that
\begin{align*}
   &\xi(\mathscr{W}^{\text{c}})-\xi(\mathscr{A})\xi(\mathscr{W}^{\text{c}})\xi(\mathscr{A})^\top - \xi(\mathscr{B})\xi(\mathscr{B})^\top\\
   &= \xi(\mathscr{W}^\text{c}-\mathscr{A}\circledast\mathscr{W}^\text{c}\circledast\mathscr{A}^\top- \mathscr{B}\circledast\mathscr{B}^\top), 
\end{align*}
which implies that $\mathscr{W}^\text{c}$ is the solution of the T-Lyapunov equation $\mathscr{W}^\text{c}-\mathscr{A}\circledast\mathscr{W}^\text{c}\circledast\mathscr{A}^\top = \mathscr{B}\circledast\mathscr{B}^\top$. A similar argument can be
established for the observability Gramian $\mathscr{W}^{\text{o}}$.
\end{proof}

By applying the Fourier transform, the T-Lyapunov equations can be decoupled into a series of independent matrix Lyapunov equations, thereby significantly simplifying the computation of the controllability and observability Gramians. Assume that the  diagonal block matrices of $\mathscr{A}$, $\mathscr{B}$, and $\mathscr{C}$ in the Fourier domain are $\textbf{A}_j\in\mathbb{R}^{n\times n}$, $\textbf{B}_j\in\mathbb{R}^{n\times m}$, and $\textbf{C}_j\in\mathbb{R}^{l\times n}$, respectively, for $j=1,2,\dots, s$.

\begin{corollary}\label{coro:gramians}
    The controllability  Gramians $\mathscr{W}^\text{c}$ of the input-output TPDS \eqref{eq:tpds} can be computed as
    \begin{equation*}
        \mathscr{W}^\text{c} = \xi^{-1} \big( (\textbf{F}_n^* \otimes \textbf{I}_s) \texttt{blkdiag}(\textbf{W}^\text{c}_1, \dots, \textbf{W}^{\text{c}}_s)(\textbf{F}_n \otimes \textbf{I}_s) \big),
    \end{equation*}
    where $\textbf{W}^\text{c}_j$ are the solutions of the Lyapunov equations
    \begin{equation*}
        \textbf{W}^\text{c}_j-\textbf{A}_j\textbf{W}^\text{c}_j\textbf{A}_j^\top = \textbf{B}_j\textbf{B}_j^\top
    \end{equation*}
    for $j=1,2,\dots, s$. 
\end{corollary}
\begin{proof}
    Based on Proposition \ref{prop:lyapunov} and the properties of the block circulant operator and the discrete Fourier transform, it can be shown that
    \begin{align*}
        &\mathcal{F}\{\xi(\mathscr{W}^\text{c}-\mathscr{A}\circledast\mathscr{W}^\text{c}\circledast\mathscr{A}^\top - \mathscr{B}\circledast\mathscr{B}^\top)\}\\
        &=\mathcal{F}\{\xi(\mathscr{W}^\text{c})\} - \mathcal{F}\{\xi(\mathscr{A}\circledast\mathscr{W}^\text{c}\circledast\mathscr{A}^\top)\}-\mathcal{F}\{\xi(\mathscr{B}\circledast\mathscr{B}^\top)\}\\
        &=\mathcal{F}\{\xi(\mathscr{W}^\text{c})\}-\mathcal{F}\{\xi(\mathscr{A})\}\mathcal{F}\{\xi(\mathscr{W}^\text{c})\}\mathcal{F}\{\xi(\mathscr{A})\}^\top\\
        &-\mathcal{F}\{\xi(\mathscr{B})\}\mathcal{F}\{\xi(\mathscr{B})\}^\top.
    \end{align*}
Each matrix in the final step is block diagonal. For example, $\mathcal{F}\{\xi(\mathscr{A})\}=\texttt{blkdiag}(\textbf{A}_1,\textbf{A}_2,\dots,\textbf{A}_s)$. Therefore, the Lyapunov equation can be divided into $s$ smaller Lyapunov equations using the corresponding diagonal block matrices. The result then follows immediately. 
\end{proof}

The construction of the observability Gramain $\mathscr{W}^{\text{o}}$ can be achieved similarly.  Once the controllability and observability Gramians $\mathscr{W}^{\text{c}}$ and $\mathscr{W}^{\text{o}}$ are computed, the next step in T-BT  involves constructing a Hankel tensor. The objective is to decompose $\mathscr{W}^{\text{c}}$ and $\mathscr{W}^{\text{o}}$ such that
\begin{equation*}
    \mathscr{W}^{\text{c}} = \mathscr{Z}_{\text{c}}\circledast\mathscr{Z}_{\text{c}}^\top \text{ and } \mathscr{W}^{\text{o}} = \mathscr{Z}_{\text{o}}\circledast\mathscr{Z}_{\text{o}}^\top.
\end{equation*}
To achieve such decompositions, we exploit T-SVD. For the controllability Gramian, since $\mathscr{W}^{\text{c}}$ is T-symmetric, the T-SVD of $\mathscr{W}^{\text{c}}$ can be computed as $\mathscr{W}^{\text{c}}=\mathscr{U}_\text{c}\circledast\mathscr{S}_\text{c}\circledast\mathscr{U}_{\text{c}}^\top$, and the resulting factor tensor can be computed as $\mathscr{Z}_{\text{c}}=\mathscr{U}_{\text{c}}\circledast\mathscr{S}_{\text{c}}^{\frac{1}{2}}\in\mathbb{R}^{n\times p\times s}$, where $p\leq n$, and $\mathscr{S}_{\text{c}}^{\frac{1}{2}}$ is defined as
\begin{equation*}
    \mathscr{S}_{\text{c}}^{\frac{1}{2}} = \xi^{-1}\big(\mathcal{F}^{-1}\{\mathcal{F}\{\xi(\mathscr{S}_{\text{c}})\}^{\frac{1}{2}}\}\big).
\end{equation*}
Suppose that $\mathscr{Z}_{\text{o}}\in\mathbb{R}^{n\times q\times s}$ is obtained similarly from the T-SVD of $\mathscr{W}^{\text{o}}$. The Hankel tensor of the input-output TPDS \eqref{eq:tpds} is then defined as $\mathscr{H}=\mathscr{Z}_{\text{o}}^\top\circledast\mathscr{Z}_{\text{c}}\in\mathbb{R}^{q\times p\times s}$.

A key property of T-SVD is the optimality of the truncated T-SVD for approximation \cite{kilmer2013third,kilmer2011factorization,zhang2018randomized}. In other words, truncating the singular tuples based on their Frobenius norm can yield the best approximation of the original tensor. Consequently, we apply T-SVD  to the Hankel tensor $\mathscr{H}$ and truncate its singular tuples based on their Frobenius norm to effectively reduce the state dimensionality. Suppose that the truncated T-SVD of $\mathscr{H}$ is given by $\mathscr{H}\approx\mathscr{U}\circledast\mathscr{S}\circledast\mathscr{V}^\top$ where $\mathscr{U}\in\mathbb{R}^{q\times r\times s}$, $\mathscr{S}\in\mathbb{R}^{r\times r\times s}$, and $\mathscr{V}\in\mathbb{R}^{p\times r\times s}$ for $r<n$. The transformation tensors therefore can be constructed as
\begin{equation}
\begin{split}
    \mathscr{P} &= \mathscr{Z}_{\text{c}}\circledast\mathscr{V}\circledast\mathscr{S}^{-\frac{1}{2}}\in\mathbb{R}^{n\times r\times s},\\
    \mathscr{Q} &= \mathscr{Z}_{\text{o}}\circledast\mathscr{U}\circledast\mathscr{S}^{-\frac{1}{2}}\in\mathbb{R}^{n\times r\times s},
\end{split}
\end{equation}
where $\mathscr{S}^{-\frac{1}{2}}$ is defined similarly as $\mathscr{S}^{\frac{1}{2}}$. The reduced input-output TPDS parameter tensors are then computed as 
\begin{equation}\label{eq:redtpds}
    \begin{split}
        \mathscr{A}_{\text{red}}&=\mathscr{Q}^\top\circledast\mathscr{A}\circledast\mathscr{P}\in\mathbb{R}^{r\times r\times s},\\
        \mathscr{B}_{\text{red}}&=\mathscr{Q}^\top\circledast\mathscr{B}\in\mathbb{R}^{r\times m\times s},\\
        \mathscr{C}_{\text{red}}&=\mathscr{C}\circledast\mathscr{P}\in\mathbb{R}^{l\times r\times s}.
    \end{split}
\end{equation}

The detailed steps of T-BT are summarized in Algorithm \ref{alg:tbt}. Most steps in the algorithm, including Steps 2 through 7, can be efficiently implemented in the Fourier domain by leveraging the diagonal block matrices of $\mathscr{A}$, $\mathscr{B}$, and $\mathscr{C}$. The final reduced TPDS parameter tensors $\mathscr{A}_{\text{red}}$, $\mathscr{B}_{\text{red}}$, and $\mathscr{C}_{\text{red}}$ can be obtained by applying the inverse Fourier transform on their corresponding diagonal block matrices. Moreover, it can be readily shown that T-BT preserves controllability and observability, ensuring that the reduced system maintains both properties. Furthermore, compared to standard BT for input-output TPDSs, T-BT offers several key advantages, including reduced memory consumption, lower computational complexity, and comparable error bounds.

\begin{algorithm}[t]
\caption{T-Balanced Truncation}
\label{alg:tbt}
\begin{algorithmic}[1]
\STATE Given the input-output TPDS \eqref{eq:tpds} with $\mathscr{A}$, $\mathscr{B}$, and $\mathscr{C}$ \\
\STATE Compute the controllability and observability Gramians
 $\mathscr{W}^{\text{c}}$ and $\mathscr{W}^{\text{o}}$ according to Corollary \ref{coro:gramians} \\
\STATE Compute the T-SVDs of $\mathscr{W}^{\text{c}}$ and $\mathscr{W}^{\text{o}}$ and obtain the factor tensors $\mathscr{Z}_{\text{c}}$ and $\mathscr{Z}_{\text{o}}$ \\
\STATE Construct the Hankel tensor $\mathscr{H}=\mathscr{Z}_{\text{o}}^\top\circledast\mathscr{Z}_{\text{c}}$ \\
\STATE Compute the truncated T-SVD of $\mathscr{H}$, i.e., $\mathscr{H}\approx \mathscr{U}*\mathscr{S}*\mathscr{V}^\top$, by truncating the singular tuples of $\mathscr{H}$ based on their Frobenius norm\\
\STATE Calculate the transformation tensors  $\mathscr{P} = \mathscr{Z}_{\text{c}}\circledast\mathscr{V}\circledast\mathscr{S}^{-\frac{1}{2}}$ and $\mathscr{Q}=\mathscr{Z}_{\text{o}}\circledast\mathscr{U}\circledast\mathscr{S}^{-\frac{1}{2}}$\\
\STATE  The reduced input-output TPDS with $\mathscr{A}_{\text{red}} = \mathscr{Q}^\top\circledast\mathscr{A}\circledast\mathscr{P}$,
 $\mathscr{B}_{\text{red}}=\mathscr{Q}\circledast\mathscr{B}$, and $\mathscr{C}_{\text{red}}=\mathscr{C}\circledast\mathscr{P}$
\end{algorithmic}
\end{algorithm}

\begin{remark}\label{rm:1}
Suppose that $p=q=n$ and we truncate $k$ singular tuples in the T-SVD of $\mathscr{H}$. The total number of parameters in the reduced input-output TPDS after T-BT is $(n-k)^2s+(n-k)ms+(n-k)ls$. Under the same level of truncation (i.e., truncating $k$ singular values), the total number of parameters in the reduced system using BT on the unfolded form \eqref{eq:unfoldtpds} is $(ns-k)^2+(ns-k)ms+(ns-k)ls$. More importantly, the reduced system becomes linear and cannot be transformed back into the form of TPDSs due to the loss of the block circulant structure.
\end{remark}

\begin{remark}
    Suppose that $p=q=n$. The computational complexity of T-BT comprises the following components: (i) computing the Fourier transforms of $\mathscr{A}$, $\mathscr{B}$, and $\mathscr{C}$ requires $\mathcal{O}((n^2+nm+nl)s\log{s})$; (ii) solving the T-Lyapunov equations for the controllability and observability Gramians $\mathscr{W}^{\text{c}}$ and $\mathscr{W}^{\text{o}}$ involves $\mathcal{O}(n^3s+n^2ms+n^2ls)$; (iii) building the Hankel tensor $\mathscr{H}$,  computing its T-SVD, and constructing the transformation tensor $\mathscr{P}$ and $\mathscr{Q}$ take $\mathcal{O}(n^3s)$; (iv) computing the reduced input-output TPDS (with the inverse Fourier transform) has the same complexity as (ii). Combining all components, the total computational complexity of T-BT is
    \begin{equation*}
        \mathcal{O}((n^2+nm+nl)s\log{s} + n^3s+n^2ms+n^2ls).
    \end{equation*}
In contrast, applying standard BT to the unfolded linear system \eqref{eq:unfoldtpds} results in a significantly higher complexity of $\mathcal{O}(n^3 s^3 + n^2ms^3+n^2ls^3)$.
\end{remark}

Last but not least, we can further derive an error bound for T-BT based on the truncation levels of singular tuples.  Before establishing this error bound, we  need to introduce the definitions of transfer functions and the $\mathcal{H}$-infinity
 norm in the context of input-output TPDSs, which can be defined through the block circulant operator. 
 
 \begin{definition}
     The transfer function of the input-output TPDS \eqref{eq:tpds} is defined as
     \begin{equation}
         \mathscr{G}(z) = \mathscr{C}\circledast(z\mathscr{I}-\mathscr{A})^{-1}\circledast \mathscr{B}\in\mathbb{R}^{l\times m\times s},
     \end{equation}
where $z$ is the complex frequency variable used in the Z-transform, and $\mathscr{I}\in\mathbb{R}^{n\times n\times s}$ is the T-identity tensor. 
 \end{definition}
 
 \begin{definition}
     The $\mathcal{H}$-infinity norm of the transfer function $\mathscr{G}(z)$ of the input-output TPDS \eqref{eq:tpds} is defined as
     \begin{equation}
         \|\mathscr{G}(z)\|_{\infty} = \|\xi(\mathscr{G}(z))\|_{\infty},
     \end{equation}
which corresponds to the maximum singular value of the unfolded matrix $\xi(\mathscr{G}(e^{j\omega}))$ for $\omega\in(-\pi,\pi]$.
 \end{definition}
 
\begin{proposition}
    Suppose that $\mathscr{G}(z)$ and $\mathscr{G}_{\text{red}}(z)$ are the transfer functions of the original and reduced input-output TPDSs (which are stable), respectively. The error bound by truncating $k$ singular tuples in the T-SVD of $\mathscr{H}$ after T-BT is given by
    \begin{equation}
        \|\mathscr{G}(z)-\mathscr{G}_{\text{red}}(z)\|_{\infty}\leq 2\max_{i=1,2,\dots,s}\sum_{j=n-k+1}^{n}\sigma_{j}^{(i)},
    \end{equation}
where $\sigma_{j}^{(i)}$ denotes the $j$the entry the $i$th singular tuple of the Hankel tensor $\mathscr{H}$ in the Fourier domain.
\end{proposition}
\begin{proof}
    Suppose that $p=q=n$ and  the complete T-SVD of $\mathscr{H}$ is provided with $\mathscr{H}=\tilde{\mathscr{U}}\circledast\tilde{\mathscr{S}}\circledast\tilde{\mathscr{V}}^\top$ where $\tilde{\mathscr{U}}\in\mathbb{R}^{n\times n\times s}$, $\tilde{\mathscr{S}}\in\mathbb{R}^{n\times n\times s}$, and $\tilde{\mathscr{V}}\in\mathbb{R}^{n\times n\times s}$. Let $\mathscr{T}=\mathscr{Z}_{\text{c}}\circledast\tilde{\mathscr{V}}\circledast\tilde{\mathscr{S}}^{-\frac{1}{2}}$. It can be shown that $\mathscr{T}^{-1}=\mathscr{Z}_{\text{o}}\circledast\tilde{\mathscr{U}}\circledast\tilde{\mathscr{S}}^{-\frac{1}{2}}$. The resulting balanced system (without truncation) parameters can be computed as
    \begin{align*}
        \tilde{\mathscr{A}}&=\mathscr{T}^{-1}\circledast\mathscr{A}\circledast\mathscr{T}\in\mathbb{R}^{n\times n\times s},\\
        \tilde{\mathscr{B}}&=\mathscr{T}^{-1}\circledast\mathscr{B}\in\mathbb{R}^{n\times m\times s},\\
        \tilde{\mathscr{C}}&=\mathscr{C}\circledast\mathscr{T}\in\mathbb{R}^{l\times n\times s},
    \end{align*}
and the associated controllability and observability Gramians are given by $\tilde{\mathscr{W}}^{\text{c}}=\tilde{\mathscr{W}}^{\text{o}}=\tilde{\mathscr{S}}\in\mathbb{R}^{n\times n\times s}$. Clearly, $\mathscr{G}(z)$ is also the transfer function for the balanced system, i.e., 
\begin{equation*}
    \mathscr{G}(z) = \tilde{\mathscr{C}}\circledast(z\mathscr{I}-\tilde{\mathscr{A}})^{-1}\circledast\tilde{\mathscr{B}} =\mathscr{C}\circledast(z\mathscr{I}-\mathscr{A})^{-1}\circledast\mathscr{B}.
\end{equation*}
Partition the F-diagonal Gramian tensor $\tilde{\mathscr{S}}$ as 
\begin{equation*}
    \tilde{\mathscr{S}} = \begin{bmatrix}
        \tilde{\mathscr{S}}_{1} & \mathscr{O}\\
        \mathscr{O} & \tilde{\mathscr{S}}_{2}
    \end{bmatrix},
\end{equation*}
where $\tilde{\mathscr{S}}_1\in\mathbb{R}^{(n-k)\times (n-k)\times s}$ and $\mathscr{O}$ denotes the zero tensor. Similarly, partition the balanced system parameters as
\begin{equation*}
    \tilde{\mathscr{A}} = \begin{bmatrix}
        \tilde{\mathscr{A}}_{11} & \tilde{\mathscr{A}}_{12}\\
        \tilde{\mathscr{A}}_{21} & \tilde{\mathscr{A}}_{22}
    \end{bmatrix}, \text{ }
    \tilde{\mathscr{B}} = \begin{bmatrix}
        \tilde{\mathscr{B}}_{1}\\
        \tilde{\mathscr{B}}_2
    \end{bmatrix}, \text{ }
    \tilde{\mathscr{C}} = \begin{bmatrix}
        \tilde{\mathscr{C}}_1 & \tilde{\mathscr{C}}_2
    \end{bmatrix},
\end{equation*}
where $\tilde{\mathscr{A}}_{11}\in\mathbb{R}^{(n-k)\times (n-k)\times s}$, $\tilde{\mathscr{B}}_1\in\mathbb{R}^{(n-k)\times m\times s}$, and $\tilde{\mathscr{C}}_1\in\mathbb{R}^{l\times (n-k)\times s}$. The difference between the transfer functions of the full and reduced systems can be computed as 
\begin{align*}
    \mathscr{G}(z) - \mathscr{G}_{\text{red}}(z) & = \begin{bmatrix}
        \tilde{\mathscr{C}}_1 & \tilde{\mathscr{C}}_2
    \end{bmatrix}
    \circledast \begin{bmatrix}
        z\mathscr{I}-\tilde{\mathscr{A}}_{11} & -\tilde{\mathscr{A}}_{12}\\
        -\tilde{\mathscr{A}}_{21} & z\mathscr{I}-\tilde{\mathscr{A}}_{22}
    \end{bmatrix}\\ &\circledast\begin{bmatrix}
        \tilde{\mathscr{B}}_1\\
        \tilde{\mathscr{B}}_2
    \end{bmatrix}-\tilde{\mathscr{C}}_1\circledast(z\mathscr{I}-\tilde{\mathscr{A}}_{11})^{-1}\circledast\tilde{\mathscr{B}}_1\\
    & = \hat{\mathscr{C}}(z)\circledast(z\mathscr{I}-\hat{\mathscr{A}}(z))^{-1}\circledast\hat{\mathscr{B}}(z),
\end{align*}
where 
\begin{align*}
    \hat{\mathscr{A}}(z) & = \tilde{\mathscr{A}}_{22}+\tilde{\mathscr{A}}_{21}\circledast(z\mathscr{I}-\tilde{\mathscr{A}}_{11})^{-1}\circledast\tilde{\mathscr{A}}_{12}\in\mathbb{R}^{k\times k\times s},\\
    \hat{\mathscr{B}}(z) &= \tilde{\mathscr{B}}_2+\tilde{\mathscr{A}}_{21}\circledast(z\mathscr{I}-\tilde{\mathscr{A}}_{11})^{-1}\circledast\tilde{\mathscr{B}}_1\in\mathbb{R}^{k\times m\times s},\\
    \hat{\mathscr{C}}(z)&=\tilde{\mathscr{C}}_1\circledast(z\mathscr{I}-\tilde{\mathscr{A}}_{11})^{-1}\circledast\tilde{\mathscr{A}}_{12} + \tilde{\mathscr{C}}_2\in\mathbb{R}^{l\times k\times s}.
\end{align*}
It can be shown that $\tilde{\mathscr{S}}_2$ satisfies the T-Lyapunov equations associated with $\hat{\mathscr{A}}(z)$, $\hat{\mathscr{B}}(z)$, and $\hat{\mathscr{C}}(z)$. By decomposing the tensor problem into $s$ smaller matrix problems in the Fourier domain, it follows that
\begin{align*}
   &\mathcal{F}\{\xi(\hat{\mathscr{C}}(z)\circledast(z\mathscr{I}-\hat{\mathscr{A}}(z))^{-1}\circledast\hat{\mathscr{B}}(z))\}\\
   & = \mathcal{F}\{\xi(\hat{\mathscr{C}}(z))\}\mathcal{F}\{\xi((z\mathscr{I}-\hat{\mathscr{A}}(z))^{-1})\}\mathcal{F}\{\xi(\hat{\mathscr{B}}(z))\}\\
    &=\begin{bmatrix}
        \hat{\textbf{C}}_1(z\textbf{I}-\hat{\textbf{A}}_1)^{-1}\hat{\textbf{B}}_1  & \cdots & \textbf{0}\\
        \vdots & \ddots & \vdots \\
        \textbf{0}  & \cdots & \hat{\textbf{C}}_s(z\textbf{I}-\hat{\textbf{A}}_s)^{-1}\hat{\textbf{B}}_s
     \end{bmatrix},
\end{align*}
where $\hat{\textbf{A}}_j(z)$, $\hat{\textbf{B}}_j(z)$, $\hat{\textbf{C}}_j(z)$ are the block diagonal matrices of $\hat{\mathscr{A}}(z)$, $\hat{\mathscr{B}}(z)$, and $\hat{\mathscr{C}}(z)$ in the Fourier domain, respectively, for $j=1,2,\dots,s$. Therefore, by applying the BT error bound from \cite{al1987error} to each diagonal entry, and using the fact that the $\mathcal{H}$-infinity norm of a block-diagonal transfer function matrix equals the maximum $\mathcal{H}$-infinity norm of its individual blocks, it follows that
\begin{align*}
    &\|\mathscr{G}(z) - \mathscr{G}_{\text{red}}(z)\|_{\infty} =  \|\mathcal{F}\{\xi(\mathscr{G}(z) - \mathscr{G}_{\text{red}}(z))\}\|_{\infty}\\
    &=\max_{i=1,2,\dots,s} \|\hat{\textbf{C}}_i(z)(z\textbf{I}-\hat{\textbf{A}}_i(z))^{-1}\hat{\textbf{B}}_i(z)\|_{\infty}\\
    &\leq 2 \max_{i=1,2,\dots,s} \sum_{j=n-k+1}^n \sigma_j^{(i)},
\end{align*}
and the proof is completed. 
\end{proof}

In summary, T-BT  exploits the intrinsic tensor structure of input-output TPDSs, thereby avoiding computationally expensive large-scale matrix operations. It reformulates the problem in the tensor domain using advanced T-SVD to enable efficient low tensor rank approximations, which is fundamentally different from BT  as T-SVD is not equivalent to SVD under the block circulant operator.  Furthermore, T-BT  provides an explicit error bound on the approximation, which ensures that the reduced model accurately preserves the key dynamic characteristics of the original system. Its ability to efficiently handle input-output TPDSs makes T-BT particularly well-suited for capture the dynamics of image and video data. Next, we introduce two variants of T-BT that utilize data.

\subsection{T-Balanced Proper Orthogonal Decomposition}
Rather than solving the T-Lyapunov equations or the corresponding matrix Lyapunov equations in the Fourier domain, one may directly approximate the controllability and observability Gramians from data or numerical simulations, referred to as empirical Gramians. This approach is particularly advantageous for large-scale input-output TPDSs, where solving the T-Lyapunov equations can be computationally prohibitive.  To obtain the empirical controllability Gramian, impulse response simulations of the input TPDS  must be conducted. Suppose that the control tensor $\mathscr{B}$ can be rewritten as $\mathscr{B}=[\mathscr{B}_1 \text{ }\mathscr{B}_2 \text{ }\cdots\text{ }\mathscr{B}_m]$. The state responses for the input TPDS  over $t= 0,1,\dots,T$ can be constructed as
\begin{equation*}
    \mathscr{X}^{(j)} = \begin{bmatrix}
        \mathscr{B}_j & \mathscr{A}\circledast\mathscr{B}_j & \cdots & \mathscr{A}^T\circledast\mathscr{B}_j
    \end{bmatrix}
\end{equation*}
for $j=1,2,\dots,m$. These snapshots are then arranged as
\begin{equation*}
    \mathscr{X} = \begin{bmatrix}
        \mathscr{X}^{(1)} &  \mathscr{X}^{(2)} & \cdots &  \mathscr{X}^{(m)}
    \end{bmatrix}\in\mathbb{R}^{n\times m(T+1)\times s},
\end{equation*}
and the empirical controllability Gramian is estimated as $\mathscr{W}^\text{c} = \mathscr{X}\circledast\mathscr{X}^\top\in\mathbb{R}^{n\times n\times s}$. Note that this controllability Gramian is consistent with the definition of BPOD after unfolding.  

\begin{lemma}\label{lm:5}
    The block circulant matrix $\xi(\mathscr{W}^\text{c})$ represents the empirical controllability Gramian of the input linear system with the state transition matrix $\xi(\mathscr{A})$ and control matrix $\xi(\mathscr{B})$. 
\end{lemma}
\begin{proof}
    Based on the properties of the block circulant operator, it can be shown that
    \begin{align*}
    \xi(\mathscr{X}^{(j)}) &= \begin{bmatrix}
        \xi(\mathscr{B}_j) & \xi(\mathscr{A})\xi(\mathscr{B}_j) & \cdots & \xi(\mathscr{A})^T\xi(\mathscr{B}_j)
    \end{bmatrix}\textbf{P}_1,\\
    \xi(\mathscr{X}) &= \begin{bmatrix}
        \xi(\mathscr{X}^{(1)}) &  \xi(\mathscr{X}^{(2)}) & \cdots &  \xi(\mathscr{X}^{(m)})
    \end{bmatrix}\textbf{P}=\textbf{X}\textbf{P}_1\textbf{P}_2,\\
        \xi(\mathscr{W}^\text{c}) &=\xi(\mathscr{X})\xi(\mathscr{X})^\top=\textbf{X}\textbf{P}_1\textbf{P}_2\textbf{P}_2^\top\textbf{P}_1^\top\textbf{X}^\top=\textbf{X}\textbf{X}^\top,
    \end{align*}
where $\textbf{P}_1$ and $\textbf{P}_2$ are appropriate column permutation matrices. Therefore, the result follows immediately. 
\end{proof}

\begin{algorithm}[t]
\caption{T-Balanced Proper Orthogonal Decomposition}
\label{alg:tbpod}
\begin{algorithmic}[1]
\STATE  Given the snapshots of state impulse responses $\mathscr{X}_t$ and $\mathscr{Y}_t$, integers $T$ and $L$, and the input-output TPDS \eqref{eq:tpds} with $\mathscr{A}$, $\mathscr{B}$, and $\mathscr{C}$ \\
\STATE   Construct the input snapshot tensor  
 \begin{align*}
   \mathscr{X}= \begin{bmatrix} 
   \mathscr{X}^{(1)} & \mathscr{X}^{(2)}  &\cdots &\mathscr{X}^{(m)} \end{bmatrix},
\end{align*} \\
where $\mathscr{X}^{(j)} = [\mathscr{X}_0\text{ } \mathscr{X}_1 \text{ }\cdots\text{ }\mathscr{X}_T]$ for $j=1,2,\dots,m$
\STATE Construct the output snapshot tensor $\mathscr{Y}$ using
 the adjoint state responses $\mathscr{Y}_t$ similarly\\
\STATE Compute the generalized Hankel tensor $\mathscr{H}=\mathscr{Y}^\top \circledast\mathscr{X}$\\
\STATE Compute the truncated T-SVD of $\mathscr{H}$, i.e., $\mathscr{H}\approx \mathscr{U}\circledast\mathscr{S}\circledast\mathscr{V}^\top$, by truncating the singular tuples of $\mathscr{H}$ based on their Frobenius norm\\
\STATE Calculate the transformation tensors  $\mathscr{P} = \mathscr{X}\circledast\mathscr{V}\circledast\mathscr{S}^{-\frac{1}{2}}$ and $\mathscr{Q}=\mathscr{Y}\circledast\mathscr{U}\circledast\mathscr{S}^{-\frac{1}{2}}$\\
\STATE  The reduced input-output TPDS with $\mathscr{A}_{\text{red}} = \mathscr{Q}^\top\circledast\mathscr{A}\circledast\mathscr{P}$,
 $\mathscr{B}_{\text{red}}=\mathscr{Q}\circledast\mathscr{B}$, and $\mathscr{C}_{\text{red}}=\mathscr{C}\circledast\mathscr{P}$
\end{algorithmic}
\end{algorithm}

The procedure proceeds similarly for constructing the output snapshot tensor $\mathscr{Y}\in\mathbb{R}^{n\times l(L+1)\times s}$ from the simulations of the adjoint TPDS defined as
\begin{equation}
    \mathscr{Y}(t+1) = \mathscr{A}^\top\circledast\mathscr{Y}(t) + \mathscr{C}^\top\circledast\mathscr{V}(t),
\end{equation}
over $t= 0,1,\dots,L$. The empirical observability Gramian can be approximated by $\mathscr{W}^{\text{o}}=\mathscr{Y}\circledast\mathscr{Y}^\top\in\mathbb{R}^{n\times n\times s}$.

The goal of T-BPOD is to obtain an approximate T-BT that remains computationally tractable for large-scale input-output TPDSs by deriving the balancing transformation from a collection of input and output snapshots. After obtaining the input and output snapshots $\mathscr{X}$ and $\mathscr{Y}$, the generalized Hankel tensor is then defined as $\mathscr{H} = \mathscr{Y}^\top \circledast \mathscr{X} \in \mathbb{R}^{l(L+1) \times m(T+1) \times s}$. The subsequent steps of T-BPOD, including computing the truncated T-SVD of $\mathscr{H}$, constructing the transformation tensors $\mathscr{P}$ and $\mathscr{Q}$, and obtaining the reduced input-output TPDS with parameter tensors \eqref{eq:redtpds}, follow a similar procedure as in T-BT. The detailed steps of T-BPOD are outlined in Algorithm \ref{alg:tbpod}, where most steps can be efficiently implemented in the Fourier domain, and the final reduced model is obtained by applying the inverse Fourier transform.

T-BPOD circumvents the need for solving the T-Lyapunov equations through numerical simulations, which are often more computationally efficient than T-BT, especially for large-scale systems, while maintaining similar memory usage. Moreover, T-BPOD provides significant memory and computational advantages over standard BPOD by leveraging the low tensor rank approximation enabled by T-SVD.

\begin{remark}
    Suppose that we truncate $k$ singular tuples in the T-SVD of $\mathscr{H}$. The total number of parameters in the reduced input-output TPDS after T-BPOD is $r^2s+rms+rls$ where $r=\min{\{l(L+1)-k, m(T+1)-k\}}$. Under the same level of truncation (i.e., truncating $k$ singular values), the total number of parameters
in the reduced system using BPOD on the unfolded form \eqref{eq:unfoldtpds} is $\hat{r}^2+\hat{r}m+\hat{r}l$ where $\hat{r} = \min{\{l(L+1)s-k,m(T+1)s-k\}}$. Similar to T-TB, the reduced linear system cannot be transformed back to the form of TPDSs. 
\end{remark}

\begin{remark}
Suppose that $mT<lL$. The computational complexity of T-BPOD can be decomposed into several components: (i) computing the Fourier transform of $\mathscr{X}$ and $\mathscr{Y}$ requires $\mathcal{O}(nlLs\log{s})$; (ii) building the generalized Hankel tensor $\mathscr{H}$, computing its T-SVD, and constructing the transformation tensor $\mathscr{P}$ and $\mathscr{Q}$ take $\mathcal{O}(nmTlLs + m^2T^2lLs)$; (iii)  computing the reduced input-output TPDS (with the inverse Fourier transform) involves $\mathcal{O}(n^2mTs+nm^2Ts+nmTls+m^2T^2s\log{s} + mTls\log{s})$. Thus, the total time complexity of T-BPOD is
\begin{align*}
    &\mathcal{O}((nlL+m^2T^2 + mTl)s\log{s}+nmTlLs\\ &+m^2T^2lLs+n^2mTs+nm^2T^2s).
\end{align*}
However, applying standard BPOD to the unfolded representation \eqref{eq:unfoldtpds} incurs a significantly higher computational cost with a time complexity
\begin{equation*}
    \mathcal{O}(nmTlLs^3+m^2T^2lLs^3+n^2mTs^3+nm^2Ts^3).
\end{equation*}
\end{remark}

\subsection{T-Eigensystem Realization Algorithm}
T-BPOD requires access to both the input-output TPDS \eqref{eq:tpds} and its adjoint, making it unsuitable for experimental settings. ERA-based approaches can overcome this limitation while delivering comparable performance, further reducing computational costs. The core idea of  T-ERA is to construct a generalized Hankel
tensor using impulse response simulations or experiments without having access to the input-output TPDS \eqref{eq:tpds}. Similar to  ERA,  we need to collect the snapshots of  impulse
responses from simulations or experiments and form the generalized Hankel tensor defined as
\begin{equation*}
\mathscr{H}=
\begin{bmatrix}
    \mathscr{Z}_0 & \mathscr{Z}_1 & \cdots & \mathscr{Z}_T\\
    \mathscr{Z}_1 & \mathscr{Z}_2 & \cdots & \mathscr{Z}_{T+1}\\
    \vdots & \vdots & \ddots & \vdots\\
    \mathscr{Z}_L & \mathscr{Z}_{L+1} & \cdots & \mathscr{Z}_{L+T}\\
\end{bmatrix}\in\mathbb{R}^{l(L+1)\times m(T+1)\times s},
\end{equation*}
where $\mathscr{Z}_j=\mathscr{C}\circledast\mathscr{A}^j\circledast\mathscr{B}\in\mathbb{R}^{l\times m\times s}$ are referred to as the Markov parameters, and $T+L+2$ is the number of snapshots. Note that this generalized Hankel tensor is consistent with the definition for  ERA after unfolding under appropriate row and column permutations. 

% \begin{lemma}
%     The block circulant matrix $\xi(\mathscr{H})$ represents the generalized Hankel tensor from  ERA based on the linear system with the state transition matrix $\xi(\mathscr{A})$, control matrix $\xi(\mathscr{B})$, and the output matrix $\xi(\mathscr{C})$ under appropriate row and column permutations.
% \end{lemma}
% \begin{proof}
%     The proof follows similarly as Lemma \ref{lm:5} by utilizing the properties of the block circulant operator. 
% \end{proof}

\begin{table*}[t]
\centering
\caption{Computational time, memory usage, and relative error comparisons between T-BT and BT for model order reduction of input-output TPDSs at various singular tuple/value truncation levels.}
\renewcommand{\arraystretch}{1.2}
\setlength{\tabcolsep}{12pt}
\small
\begin{tabular}{c  c c  c c  c c}
\hline
\multirow{1}{*}{Truncation Level \( k \)} & \multicolumn{2}{c}{{Time (in seconds)}} & \multicolumn{2}{c}{{Memory Usage (in MB)}} & \multicolumn{2}{c}{{Relative Error}} \\  
%\cline{2-3} \cline{4-5} \cline{6-7}
& T-BT & BT & T-BT & BT & T-BT & BT \\  
\hline
55 & 0.4210  & 1.4041 & 0.1737   & 1.3257  & \(1.2 \times 10^{-14}\)  & \(1.2 \times 10^{-14}\)  \\  
60 & 0.4166  & 1.3528  & 0.1431   & 1.0503  & \(1.09 \times 10^{-13}\)  & \(9.9 \times 10^{-14}\)  \\  
65 & 0.3833  & 1.3499  & 0.1161   & 0.8073  & \(6.38 \times 10^{-12}\) & \(6.37 \times 10^{-12}\) \\  
70 & 0.4068  & 1.3057  & 0.0927   & 0.5967  & \(3.74 \times 10^{-10}\)  & \(3.74 \times 10^{-10}\)  \\  
75 & 0.3735  & 1.4522 & 0.0729  & 0.4185  & \(2.81 \times 10^{-8}\)  & \(2.81 \times 10^{-8}\)  \\  
80 & 0.3837  & 1.4601  & 0.0567   & 0.2727  & \(1.05 \times 10^{-6}\)  & \(1.05 \times 10^{-6}\)  \\  
85 & 0.3675  & 1.3402  & 0.0441   & 0.1593  & \(7.06 \times 10^{-5}\)  & \(6.26 \times 10^{-5}\)  \\  
90 & 0.3744  & 1.3217  & 0.0351   & 0.0783  & \(2.19 \times 10^{-3}\)  & \(2.19 \times 10^{-3}\)  \\  
% 95 & 0.3875  & 1.3780  & 0.0297   & 0.0297  & \(8.88 \times 10^{-2}\)  & \(8.88 \times 10^{-2}\)  \\  
% 100 & 0.3875  & 1.3780  & 0.0297   & 0.0297  & \(1.00 \)  & \(1.00 \)  \\  
\hline
\end{tabular}
\label{tab:tbt}
\end{table*}

After obtaining the generalized Hankel tensor $\mathscr{H}$, we can compute its truncated T-SVD, expressed as $\mathscr{H} \approx \mathscr{U} \circledast \mathscr{S} \circledast \mathscr{V}^\top$, where $\mathscr{U} \in \mathbb{R}^{l(L+1) \times r \times s}$, $\mathscr{S} \in \mathbb{R}^{r \times r \times s}$, and $\mathscr{V} \in \mathbb{R}^{m(T+1) \times r \times s}$. Finally, the reduced input-output TPDS parameter tensors can be computed as
\begin{equation}
    \begin{split}
        \mathscr{A}_{\text{red}} &= \mathscr{S}^{-\frac{1}{2}}\circledast \mathscr{U}^\top\circledast\hat{\mathscr{H}}\circledast \mathscr{V}\circledast\mathscr{S}^{-\frac{1}{2}}\in\mathbb{R}^{r\times r\times s},\\
        \mathscr{B}_{\text{red}} &= \mathscr{S}^{-\frac{1}{2}}\circledast \mathscr{U}^\top \circledast \mathscr{H}_{\text{col}}\in\mathbb{R}^{r\times m\times s},\\
        \mathscr{C}_{\text{red}} &= \mathscr{H}_{\text{row}} \circledast\mathscr{V}\circledast\mathscr{S}^{-\frac{1}{2}}\in\mathbb{R}^{l\times r\times s},
    \end{split}
\end{equation}
where $\mathscr{H}_{\text{col}}$ and $\mathscr{H}_{\text{row}}$ denote the first block column and row of $\mathscr{H}$, respectively, and 
\begin{equation*}
   \hat{\mathscr{H}}  = \begin{bmatrix} \mathscr{Z}_1 & \mathscr{Z}_2 &\cdots &\mathscr{Z}_{T+1} \\ \mathscr{Z}_2 & \mathscr{Z}_3 & \cdots &\mathscr{Z}_{T+2} \\ \vdots &\vdots &\ddots &\vdots\\ \mathscr{Z}_{L+1}&\mathscr{Z}_{L+2}&\cdots &\mathscr{Z}_{T+L+1}  \end{bmatrix}\in\mathbb{R}^{l(L+1)\times m(T+1)\times s}.
\end{equation*}

The detailed steps of  T-ERA are summarized in Algorithm \ref{alg:tera}.  T-ERA does not require prior knowledge of the original input-output TPDS \eqref{eq:tpds}, further accelerating the identification of the reduced system compared to T-BPOD. Moreover, it
provides remarkable memory and computational advantages over standard ERA.

\begin{algorithm}[t]
\caption{T-Eigensystem Realization Algorithm}
\label{alg:tera}
\begin{algorithmic}[1]
\STATE  Given the snapshots of  impulse responses $\mathscr{Z}_j\in \mathbb{R}^{l\times m \times s}$, and integers $T$  and $L$ \\
\STATE  Construct the generalized Hankel tensors $\mathscr{H}$ and $\hat{\mathscr{H}}$ using the snapshots of  impulse responses $\mathscr{Z}_j$ for $j=0,1,\dots,T+L+1$\\
\STATE Compute the truncated T-SVD of $\mathscr{H}$, i.e., $\mathscr{H} \approx \mathscr{U} \circledast \mathscr{S} \circledast \mathscr{V}^\top$, by truncating the singular tuples of  $\mathscr{H}$ based on their Frobenius norm\\
\STATE The reduced input-output TPDS is then computed as
\begin{align*}
\mathscr{A}_{\text{red}} &= \mathscr{S}^{-\frac{1}{2}}\circledast \mathscr{U}^\top\circledast\hat{\mathscr{H}}\circledast \mathscr{V}\circledast\mathscr{S}^{-\frac{1}{2}},\\
        \mathscr{B}_{\text{red}} &= \mathscr{S}^{-\frac{1}{2}}\circledast \mathscr{U}^\top \circledast \mathscr{H}_{\text{col}},\\
        \mathscr{C}_{\text{red}} &= \mathscr{H}_{\text{row}} \circledast\mathscr{V}\circledast\mathscr{S}^{-\frac{1}{2}},
\end{align*}
where $\mathscr{H}_{\text{col}}$ and $\mathscr{H}_{\text{row}}$ denote the first block column and row of $\mathscr{H}$, respectively. 
\end{algorithmic}
\end{algorithm}

\begin{remark}
Suppose we truncate $k$ singular tuples in the T-SVD of $\mathscr{H}$. The total number of parameters in the reduced input-output TPDS after T-ERA remains the same as in T-BPOD. More importantly,  T-ERA can directly identify the reduced input-output TPDS, whereas traditional ERA can only recover a linear system with more parameters.
\end{remark}

\begin{remark}
Suppose that $mT<lL$. The computational complexity of  T-ERA  consists of  several components: (i) calculating the Fourier transform of the generalized Hankel tensor $\mathcal{H}$ requires $\mathcal{O}(mTlLs\log{s})$; (ii) obtaining the T-SVD of the generalized Hankel tensor $\mathcal{H}$ takes $\mathcal{O}(m^2T^2lLs)$; (iii) computing the reduced input-output TPDS (with the inverse Fourier transform) involves $\mathcal{O}(m^2T^2lLs+m^2T^2s\log{s} + mTls\log{s})$. Hence, the overall computational complexity of the T-ERA is
\begin{equation*}
    \mathcal{O}(mTlLs\log{s}+m^2T^2lLs).
\end{equation*}
On the other hand,  standard ERA has a significantly higher complexity of $\mathcal{O}(m^2T^2lLs^3)$. 
\end{remark}

\section{Numerical Examples}\label{sec:num}
All numerical experiments were conducted using MATLAB R2021b on a machine equipped with an M1 Pro CPU and 16 GB of memory, utilizing the tensor-tensor-product-toolbox-master \cite{lu2018tensor}. The code used for these experiments is
available at \url{https://github.com/ZiqinHe/TPDS_MOR}.

\begin{table*}[t]
\centering
\caption{Computational time, memory usage, and relative error comparisons between T-BPOD and BPOD for model order reduction of input-output TPDSs at various singular tuple/value truncation levels.}
\renewcommand{\arraystretch}{1.2}
\setlength{\tabcolsep}{12pt}
\small
\begin{tabular}{c c c c c c c}
\hline
\multirow{1}{*}{Truncation Level \( k \)} & \multicolumn{2}{c}{Time (in seconds)} & \multicolumn{2}{c}{Memory Usage (in MB)} & \multicolumn{2}{c}{Relative Error} \\  
& T-BPOD & BPOD & T-BPOD & BPOD & T-BPOD & BPOD \\  
\hline
55  & 0.4575  & 3.7431  & 0.2066  & 1.6348  & $4.0 \times 10^{-15}$  & $6.0 \times 10^{-15}$  \\  
60  & 0.4683  & 3.6727  & 0.1724  & 1.3270  & $9.5 \times 10^{-14}$  & $9.5 \times 10^{-14}$  \\  
65  & 0.3249  & 4.4751  & 0.1418  &  1.0516  & $5.5 \times 10^{-12}$ & $5.5 \times 10^{-12}$ \\  
70  & 0.3090  & 3.6284  &  0.1148  & 0.8086  & $2.6 \times 10^{-10}$ & $2.6 \times 10^{-10}$ \\  
75  & 0.4095  & 3.7941  & 0.0914  & 0.5980   & $1.7 \times 10^{-8}$  & $1.7 \times 10^{-8}$  \\  
80  & 0.2981  & 3.7154  & 0.0716  & 0.4198   & $1.0 \times 10^{-6}$  & $1.0 \times 10^{-6}$  \\  
85  & 0.2910  & 3.6538  & 0.0554  & 0.2740   & $4.3 \times 10^{-5}$  & $4.3 \times 10^{-5}$  \\  
90  & 0.2786  & 4.3218  &  0.0428   & 0.1606   & $2.3 \times 10^{-3}$  & $2.1 \times 10^{-3}$  \\  
% 95  & 0.3256  & 3.7462  & 0.0338   &  0.0796   & $6.9 \times 10^{-2}$  & $6.9 \times 10^{-2}$  \\  
% 100 & 0.2843  & 4.3911  & 0.0284   & 0.0310   & $1.00$                & $1.00$                \\  
\hline
\end{tabular}
\label{tab:tbpod}
\end{table*}

\subsection{T-Balanced Truncation}
In this example, we evaluated the performance of T-BT for input-output TPDSs and compared it against standard BT. To set up the problem, we first randomly generated the state transition tensor $\mathscr{A}\in\mathbb{R}^{100\times 100\times 9}$, the control tensor $\mathscr{B}\in\mathbb{R}^{100\times 5\times 9}$, and the output tensor $\mathscr{C}\in\mathbb{R}^{5\times 100\times 9}$. The primary objective was to evaluate and compare T-BT and BT in terms of computational efficiency, memory usage, and relative error in constructing the reduced systems. Memory usage is measured by the total number of parameters in the reduced system, and the relative error is defined as
\begin{equation*}
    \text{Relative Error} = \frac{\|\mathcal{G}(z)-\mathcal{G}_{\text{red}}(z)\|_{\infty}}{\|\mathcal{G}(z)\|_{\infty}},
\end{equation*}
where $\mathcal{G}(z)$ and $\mathcal{G}_{\text{red}}(z)$ are the transfer functions of the original and reduced systems, respectively, to quantify the accuracy of the reduced systems from T-BT and BT. The results are presented in Table \ref{tab:tbt}, comparing the performance of T-BT and BT across different singular tuples/values truncation levels.  Notably, the computational time for T-BT remains consistently lower than that of BT across all values of $k$. Additionally, T-BT results in reduced systems that require less memory compared to those produced by  BT across all values of $k$. Despite these advantages in terms of computational time and memory usage, T-BT also achieves relative errors comparable to those produced by  BT. Therefore, T-BT emerges as a highly effective approach for model order reduction in the context of TPDSs, offering both computational and memory advantages without compromising on accuracy.

\subsection{T-Balanced Proper Orthogonal Decomposition}
In this example, we examined the effectiveness of T-BPOD for input-output TPDSs and compared it with standard BPOD. As in the previous example, we randomly generated the state transition tensor $\mathscr{A}\in\mathbb{R}^{100\times 100\times 9}$, the control tensor $\mathscr{B}\in\mathbb{R}^{100\times 5\times 9}$, and the output tensor $\mathscr{C}\in\mathbb{R}^{5\times 100\times 9}$, and manually computed the impulse responses. We then compared T-BPOD and BPOD in terms of computational time, memory usage, and relative error in constructing the reduced systems.   The results are shown in Table \ref{tab:tbpod}. Note that the computational time includes the time for computing the impulse responses.  Similar to the results obtained with T-BT, T-BPOD demonstrates lower computational time, reduced memory usage, and nearly identical relative error when compared to BPOD, underscoring the effectiveness of T-BPOD as a viable method for model order reduction in input-output TPDSs. 

\begin{table*}[t]
\centering
\caption{Computational time, memory usage, and relative error comparisons between T-ERA and ERA for model order reduction of input-output TPDSs at various singular tuple/value truncation levels.}
\renewcommand{\arraystretch}{1.2}
\setlength{\tabcolsep}{12pt}
\small
\begin{tabular}{c c c c c c c}
\hline
\multirow{1}{*}{Truncation Level \( k \)} & \multicolumn{2}{c}{Time (in seconds)} & \multicolumn{2}{c}{Memory Usage (in MB)} & \multicolumn{2}{c}{Relative Error} \\  
& T-ERA & ERA & T-ERA & ERA & T-ERA & ERA \\  
\hline
55  & 0.3819  & 1.5076  & 0.2093  & 1.6348  & $4.0 \times 10^{-15}$  & $5.0 \times 10^{-15}$  \\  

60  & 0.2433  & 1.6379  & 0.1751  & 1.3270  & $9.5 \times 10^{-14}$  & $9.5 \times 10^{-14}$  \\  

65  & 0.2777  & 1.4462  & 0.1445 & 1.0516  & $5.5 \times 10^{-12}$ & $5.5 \times 10^{-12}$ \\  

70  & 0.2504  & 1.6192  & 0.1175  & 0.8086  & $2.6 \times 10^{-10}$ & $2.6 \times 10^{-10}$ \\  

75  & 0.2459  & 1.6999  & 0.0941  & 0.5980  & $1.7 \times 10^{-8}$  & $1.7 \times 10^{-8}$  \\  

80  & 0.2323  & 1.5652  & 0.0743  & 0.4198  & $1.0 \times 10^{-6}$  & $1.0 \times 10^{-6}$  \\  

85  & 0.2433  & 1.5592  & 0.0581  & 0.2740  & $4.3 \times 10^{-5}$  & $4.3 \times 10^{-5}$  \\  

90  & 0.2292  & 1.3671  & 0.0455  & 0.1606 & $2.3 \times 10^{-3}$  & $2.1 \times 10^{-3}$  \\  

% 95  & 0.2457  & 1.5958  & 0.0365  & 0.0796  & $6.9 \times 10^{-2}$  & $6.9 \times 10^{-2}$  \\  

% 100 & 0.2241  & 1.6584  & 0.0311  & 0.0310  & $1.00$                & $1.00$                \\  
\hline
\end{tabular}
\label{tab:tera}
\end{table*}

\subsection{T-Eigensystem Realization Algorithm}
In this example, we applied T-ERA to the same input-output TPDSs previously used for evaluating T-BPOD. We manually computed the Markov parameters and constructed the generalized Hankel tensors. We then compared T-ERA and ERA across key performance metrics, including computational time (including the time for computing the Markov parameters), memory usage, and relative error in constructing the reduced systems.  The results are demonstrated in Table \ref{tab:tera}, in which    T-ERA outperforms ERA in terms of faster computational time, lower memory requirements, and nearly identical relative error.  Additionally, the computational time of T-ERA is lower than that of T-BPOD, but it results in similar memory usage and relative error. Interestingly, the final reduced models obtained from T-ERA differ slightly from those obtained from T-BPOD, as indicated by the variation in memory usage, whereas ERA and BPOD produce identical reduced models as theoretically guaranteed. This observation suggests that, unlike their matrix-based counterparts, T-BPOD and T-ERA may not be equivalent, warranting further theoretical exploration. Overall, T-ERA demonstrates strong potential as an efficient and reliable data-driven approach for model order reduction of TPDSs.

% In this example, we assessed the performance of  T-ERA for input-output TPDSs and compared it with  standard ERA. As in the previous experiments, we randomly generated the state transition tensor $\mathscr{A}\in\mathbb{R}^{100\times 100\times 9}$, the control tensor $\mathscr{B}\in\mathbb{R}^{100\times 5\times 9}$, and the output tensor $\mathscr{C}\in\mathbb{R}^{5\times 100\times 9}$, and manually calculated the impulse responses to construct the generalized Hankel tensors. We then compared T-ERA and ERA across key performance metrics, including computational time (including the time for computing the impulse responses), memory usage, and relative error in constructing the reduced systems.   The results are shown in Table \ref{tab:tera}.   Similar to the results obtained with T-BT and T-BPOD,  T-ERA outperforms ERA in terms of faster computational time, lower memory requirements, and nearly identical relative error. Additionally, the computational time of T-ERA is lower than that of T-BPOD, but it results in similar memory usage and relative error. Interestingly, the final reduced models obtained from T-ERA differ slightly from those obtained from T-BPOD, as indicated by the variation in memory usage, whereas ERA and BPOD produce identical reduced models (theoretically guaranteed). These outcomes reinforce T-ERA's potential as an efficient and reliable method for model order reduction in input-output TPDSs.

\subsection{A Case Study on Image Data}
\begin{figure}[t]
    \centering
    \includegraphics[width=\linewidth]{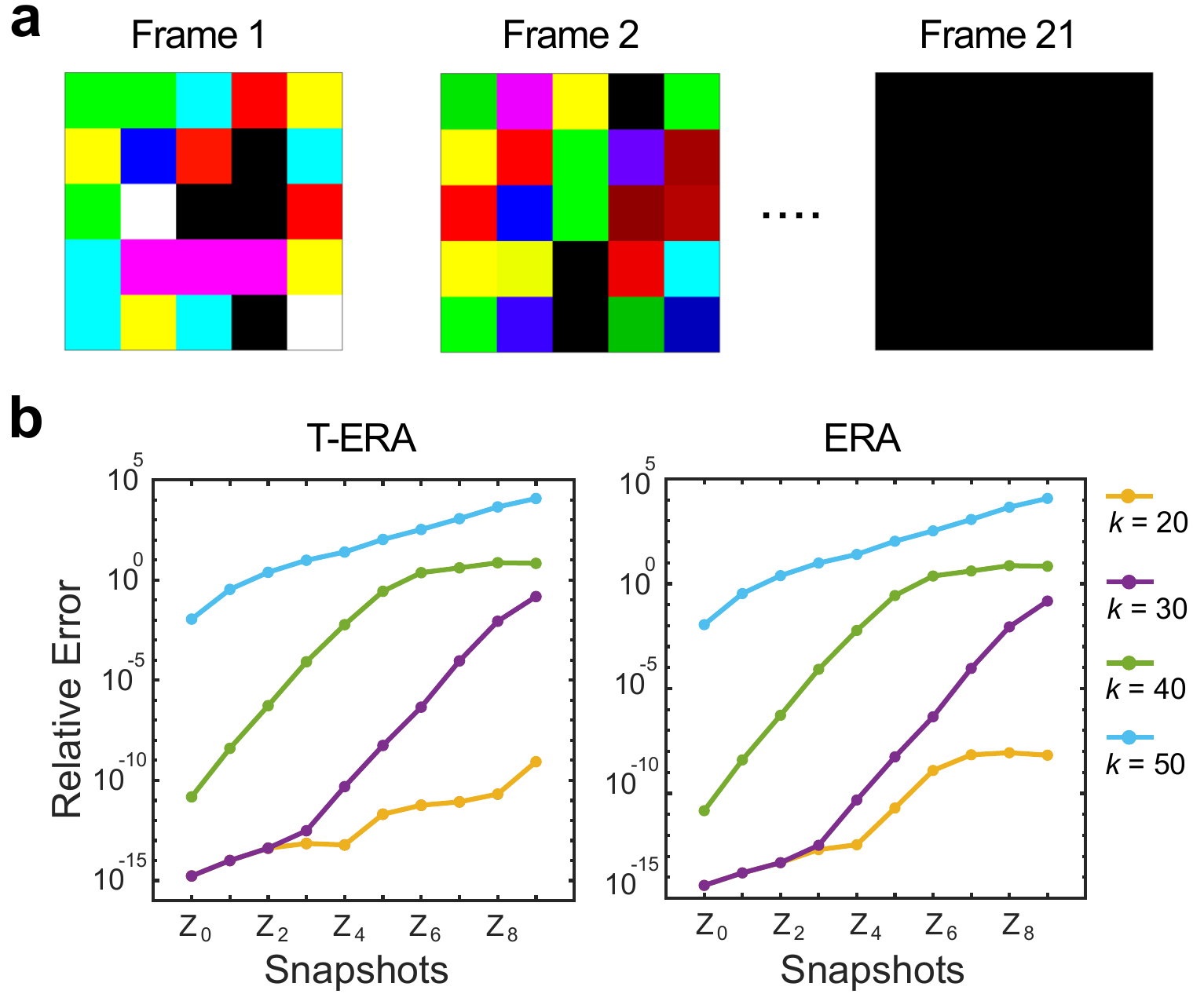}
    \caption{(a) Impulse image data containing 21 frames. (b) Relative reconstruction errors  of $\mathscr{Z}_j$ for $j=0,1,\dots,9$ using the T-ERA- and ERA-based identified systems with truncation levels $k=20,30,40,50$. The results for $k = 0$ and $k = 10$ yield similar curves to those for $k = 20$ and are omitted for clarity.}
    \label{fig:1}
\end{figure}

In this example, we applied T-ERA and ERA to a synthetic impulse image dataset (resembling color test patterns) to identify the underlying image dynamics. The dataset consists of 21 impulse images, denoted by  $\mathscr{Z}_j\in\mathbb{R}^{5\times 5\times 3}$ for $j = 0, 1, \dots, 20$ with $T=L=10$, which   were generated to progressively darken over time, ensuring that the underlying dynamics is stable, see Figure \ref{fig:1} (a). In order to apply ERA, each image frame was vectorized into a one-dimensional array, whereas T-ERA operates directly on the tensorial structure of the data. Table \ref{tab:video_data_ERA} summarizes the results, where we compared the total number of parameters in the identified models across various truncation levels of the Hankel singular tuples/values.  Notably, T-ERA exhibits a substantial memory usage advantage over ERA across all truncation levels. Additionally, the  systems recovered by T-ERA preserve a multilinear structure, taking the form of TPDSs, which naturally aligns with the tensor representation of image sequences. In contrast, ERA produces standard linear systems that flatten the underlying structure and cannot be converted back to the TPDS format. To evaluate reconstruction accuracy, we computed the relative reconstruction errors for the first 10 frames, i.e., $\mathscr{Z}_j$ for $j = 0, 1, \dots, 9$. As shown in Figure \ref{fig:1} (b), T-ERA achieves reconstruction performance comparable to that of ERA, despite using much fewer parameters. Taken together, these results demonstrate the promise of T-ERA as a powerful and memory-efficient tool for modeling the dynamics of image and video data, achieving accurate reconstructions while maintaining structural integrity and interpretability.

\begin{table}[t]
\centering
\renewcommand{\arraystretch}{1.2}
\setlength{\tabcolsep}{12pt}
\small
\caption{Total number of parameters in the identified input-output systems using T-ERA and ERA at different singular tuple/value truncation levels.}
\begin{tabular}{c c c}
\hline
\multirow{1}{*}{Truncation Level \( k \)} &  \multicolumn{2}{c}{Number of Parameters}  \\  
& T-ERA & ERA \\  
\hline
0  &  10725  & 32175    \\  

10   & 7425  & 22275   \\  

20  &  4725  & 14175    \\  

30   & 2625  & 7875   \\  

40  &  1125  & 3375    \\  

50   & 225  & 675   \\  
\hline
\end{tabular}
\label{tab:video_data_ERA}
\end{table}

\section{Conclusion}\label{sec:con}
In this article, we introduced an innovative framework for BT-based model order reduction of input-output TPDSs, where the state, input, and output variables are represented as tensors, and system evolution is described using the T-product.  Specifically, we developed advanced T-product-based methods, including T-BT, T-BPOD, and T-ERA, by leveraging the unique properties of T-SVD. These methods effectively reduce the state dimensionality of input-output TPDSs by truncating the singular tuples of their generalized Hankel tensors, while preserving the critical dynamic features of the original system. Additionally, they achieve notable improvements in computational and memory efficiency compared to traditional techniques such as BT, BPOD, and ERA.   Numerical examples highlight the effectiveness of the proposed framework, demonstrating its potential for handling large-scale tensor data while maintaining essential system dynamics.

Future work will involve applying the developed model order reduction techniques to a wide variety of real-world large-scale image and video data. The goal is to optimize the use of memory and computational time while maintaining the critical dynamic features of the data, allowing for more efficient storage and faster processing speeds. Additionally, a key direction will be incorporating advanced tensor decomposition methods, including CANDECOMP/PARAFAC decomposition \cite{kolda2009tensor}, higher-order singular value decomposition \cite{de2000multilinear}, tensor train decomposition \cite{oseledets2011tensor}, and hierarchical tucker decomposition \cite{grasedyck2010hierarchical}, to further reduce memory usage and computational time, making the proposed techniques even more scalable and efficient for large tensor datasets. Furthermore, integrating these approaches with machine learning frameworks holds great potential for enhancing automated processing tasks such as image and video compression, object detection, and motion analysis. By combining model order reduction with machine learning, we expect to improve not only computational efficiency but also the accuracy and effectiveness of data-driven solutions in real-time applications.

\section*{References}
\bibliographystyle{IEEEtran}
\bibliography{references}

\end{document}